\documentclass[letterpaper, 10 pt, conference]{ieeeconf}  

\IEEEoverridecommandlockouts                              

\usepackage{graphics} 
\usepackage{epsfig} 
\usepackage{amsmath} 
\usepackage{amsfonts}

\newtheorem{proposition}{Proposition}

\newtheorem{remark}{Remark}

\newcommand{\V}{\mathcal{V}}
\newcommand{\N}{\mathcal{N}}
\newcommand{\E}{\mathcal{E}}
\newcommand{\R}{\mathbb{R}}

\title{\LARGE \bf
Passivity-Based Control of Human-Robotic Networks with Inter-Robot Communication Delays and Experimental Verification
}

\author{J. Yamauchi$^{1}$, M.W.S. Atman$^{1}$, T. Hatanaka$^{2}$,  
N. Chopra$^{3}$ and M. Fujita$^{2}$
\thanks{$^{1}$J. Yamauchi and M.W.S. Atman are with Department of Mechanical and Control Engineering, 
		Tokyo Institute of Technology, Japan
        {\tt\small {yamauchi.j.ab, atman.m.aa}@m.titech.ac.jp}}%
\thanks{$^{2}$T. Hatanaka and M. Fujita are with Department of Systems and Control Engineering, School of Engineering
		Tokyo Institute of Technology, Japan
        {\tt\small hatanaka@ctrl.titech.ac.jp}}%
\thanks{$^{3}$N. Chopra is with Department of Mechanical Engineering,
		University of Maryland, USA}%
}

\begin{document}

\maketitle
\thispagestyle{empty}
\pagestyle{empty}

\begin{abstract}

In this paper, we present experimental studies on a cooperative control system
for human-robotic networks with inter-robot communication delays.
We first design a cooperative controller to be implemented on each robot
so that their motion are synchronized to a reference motion desired by a human operator,
and then point out that each robot motion ensures passivity.
Inter-robot communication channels are then designed via so-called scattering
transformation which is a technique to passify the delayed channel.
The resulting robotic network is then connected with human operator based on passivity theory.
In order to demonstrate the present control architecture, we build an experimental testbed
consisting of multiple robots and a tablet.
In particular, we analyze the effects
of the communication delays on the human operator's behavior.

\end{abstract}

\section{Introduction}
\label{sec:intro}

Complex robotic coordination tasks in highly uncertain environments 
require the system designers to
take advantage of human operator's strengths, 
high-level decision-making and flexibility.
Motivated by these needs, semi-autonomous operation of robots
is gaining increasing research interests \cite{WZ17}.

One of the most promising design tools for such semi-autonomous systems is
passivity, as confirmed in the history of a traditional semi-autonomous robot control problem, bilateral teleoperation \cite{RTEM10}-\cite{STY10}.
In this research field, the human operator has been treated as a passive component and 
then control architectures have been established while ensuring closed-loop stability based upon this assumption.
This paradigm has also been taken in teleoperation of multiple networked robots
\cite{RTEM10}--\cite{FSSBG12}.
Rodriguez-Seda et al. \cite{RTEM10} present an architecture such that the robotic network is 
controlled by operating a leader robot so that the robots forms a specified formation while avoiding collisions.
Liu \cite{L15} extends motion synchronization on the joint space to the task space under time-varying delays. 
Fully distributed control algorithms are presented by Franchi et al. \cite{FSSBG12}, where stability
is ensured in the presence of split and join events.
Other interesting extensions of bilateral teleoperation are presented in \cite{VZ13} and \cite{SW16}.
Varnell and Zhang \cite{VZ13} employ a non-classical human-robot interfaces, namely a tablet,
and discuss the stability while assuming a human pointing model. 
Saeidi et al. \cite{SW16} introduce the notion of robot-to-human trust and mixed initiative control scheme
in order to improve the performance and reduce workload of the operator.

While the above papers consider robot dynamics together with force feedback,
the authors \cite{HCF15,HCYF17} studied another problem formulation
focusing on interactions between a human operator and kinematic robots as in \cite{ECKK14,SMH15}. 
The schematic of the intended system is illustrated in Fig. \ref{fig:scenario}.
Then, we investigated motion synchronization of robots to human references under
distributed information exchange between the human and robots, and among robots.
To this end, we presented a novel passivity-based control architecture  and proved synchronization
under explicit bilateral connections between the human and robotic network.
For implementation on real robotic systems,
both  the human-robotic network and inter-robot interactions must be implemented
using appropriate communication technology. 
In this case, the communication channels may suffer from delays 
in the transmission of information,
which is indeed one of the most important issues in bilateral teleoperation.
However, \cite{HCF15,HCYF17} did not address this issue explicitly.

\begin{figure}[t]
	\centering
	\includegraphics[width=80mm,clip]{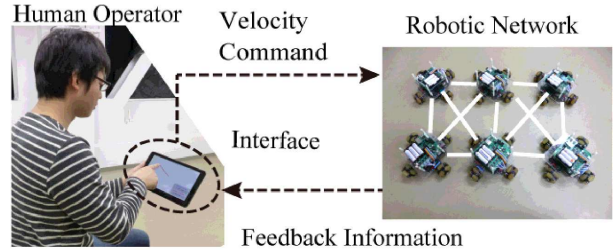}
	\caption{Information flows in the intended scenario of human-robotic 
	network system.}
	 \label{fig:scenario}
\end{figure}
%
%


Our objective is thus to extend the results of \cite{HCF15,HCYF17} 
to the case with time delays and its experimental studies.
Although both the human-robotic network and inter-robot communications may have delays,
this paper investigates the latter since this issue emerges in the one-human-multiple-robot interactions
while the former has been in-depth studied in bilateral teleoperation.
Delays in inter-robot communications have been extensively studied in the field of cooperative control
(See \cite{HCFS15} and references therein).
Among many approaches, this paper focuses on the scheme presented in 
\cite{HCFS15}, 
and synchronization is guaranteed despite the delays.
In this paper, the control architecture in \cite{HCFS15} is shown to be successfully integrated with 
the cooperative control for human-robotic networks presented in \cite{HCF15,HCYF17}. 
However, the integration is not straightforward and requires novel technical extensions in 
the formulation and this is presented in this paper.
Furthermore, as pointed out in \cite{DP16}, real deployments of robotic 
network is important from the viewpoint of difficulties
to simulate each intended task and situation faithfully.
Therefore, we investigate our proposed architecture in
a testbed and show the influences of communication delays 
not only on a robotic network but also on a real human's behaviour.

The organization of this paper is as follows: 
Section \ref{sec:problem} presents the intended problem and briefly reviews the results of \cite{HCF15,HCYF17}.
In Section \ref{sec:hrn}, we design a novel control architecture based on the scattering transformation \cite{HCFS15}, and then, 
show passivity of the resulted system.
We also show position synchronization to the reference values in Section \ref{sec:main}.
We demonstrate our results through experiments in Section \ref{sec:experiment}, 
and finally, summarize the obtained results 
in Section \ref{sec:conclusions}.

%
\section{Problem Setting}
\label{sec:problem}

In this section, we start by reviewing the problem formulation
and results presented in \cite{HCF15,HCYF17}.
Please refer to \cite{HCF15,HCYF17} for more details.

\begin{figure}[t]
	\centering
	\includegraphics[width=82mm,clip]{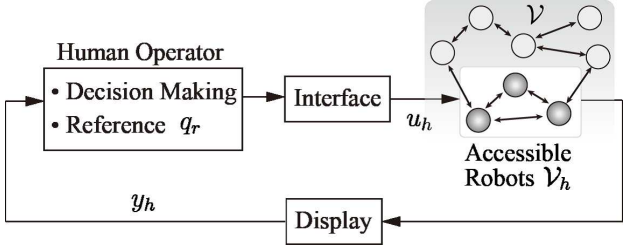}
	\caption{The system configuration for human-robotic network.}
	\label{fig:scenarioblock}
\end{figure}
%
%

\subsection{System Configuration and Objectives}

Let us consider a system with a human operator and $n$ mobile robots 
$\V = \{1, \cdots, n\}$ located on a 2-D plane as shown 
in Fig. \ref{fig:scenario}.
We suppose that the motion of every $i$-th robot is described by the kinematic model
\begin{eqnarray} \label{eq:integrator}
\dot{q}_i = u_i \quad \forall{i} \in \V,
\end{eqnarray}
where $q_{i} \in \R^{2}$ and $u_{i} \in \R^{2}$ are the position and 
velocity input of $i$-th robot, respectively.

Every robot can interact with neighboring robots and
the inter-robot information exchanges are modeled by a 
graph $G = (\V, \E),\ \E \subseteq \V \times \V$.
Then, robot $i$ has access to information of the robots belonging to 
the set $\N_{i} = \{j \in \V |\ (i,j) \in \E\}$.
Throughout this paper, 
we assume the communication graph is fixed, undirected and connected.
Differently from \cite{HCF15}, we assume that the inter-robot communication
suffers from time delays, which will be discussed 
in Section \ref{sec:problem}.

In addition to the total robots set $\V$, 
we define accessible robots set $\V_{h}$ which is a subset of $\V$ (Fig. \ref{fig:scenarioblock}).
\textit{Accessibility} means that the operator 
can send a signal to and receive feedback information from
only the robots in $\V_{h}$.
Here, we assume that 
the human operator determines a command signal $u_{h}$ based on
certain information $y_{h}$ visually fed back from $\V_{h}$ through a monitor
in front of the operator. 
The command $u_h$ is then sent back to all the accessible robots.
We also introduce the notation $\delta_{i}$ such that 
$\delta_{i} = 1$ if $i \in \V_{h}$ and $\delta_{i} = 0$ otherwise.

In this paper, we address position synchronization.
Let us assume that the human operator has a desired position denoted by $q_{r}$.
Then, the goal of the position synchronization is defined by
\begin{align}\label{eq:goal1}
\lim_{t \to \infty}\|q_i - q_r\| = 0 \quad \forall{i} \in \V.
\end{align}
The objective here is to design the robot controller $u_{i}$ and the information $y_{h}$ displayed on the monitor 
so as to guarantee the above control goal. 

\subsection{Control Architecture without Time Delays}


The authors \cite{HCF15,HCYF17} take the control input $u_{i}$ as
\begin{eqnarray}
u_{i} = u_{r,i} + \delta_{i}u_{h} \quad i \in \V.
\label{eqn:hata}
\end{eqnarray}
Then, the signal $u_{r,i}$ to achieve inter-agent motion synchronization 
is designed as
\begin{align}
\dot{\xi}_{i} 
&= 
\sum_{j \in \N_{i}} b_{ij} (q_{j} - q_{i}) \label{eq:xiorig} \\
u_{r,i}
&= 
\sum_{j \in \N_{i}} a_{ij} (q_{j} - q_{i}) 
- \sum_{j \in \N_{i}} b_{ij} (\xi_{j} - \xi_{i})
\label{eq:uriorig} 
\end{align}
where $a_{ij} = a_{ji}, b_{ij} = b_{ji}$ $\forall{i, j} \in \V$,
$a_{ij} > 0, b_{ij} > 0$ if $(i, j) \in \E$ and $a_{ij} = b_{ij} = 0$ 
otherwise.

Then, 
the collective dynamics (\ref{eq:integrator}), (\ref{eqn:hata}) 
and (\ref{eq:uriorig}) for all $i$ is passive from $u_{h}$ to $z$, 
where $z$ is defined as
\begin{eqnarray}
\label{eq:zqzv}
z := \frac{1}{m}\sum_{i \in \V_{h}}q_{i},\quad 
\end{eqnarray}
and $m$ is the number of elements of $\V_{h}$.
Based on this passivity, we let the  feedback information $y_{h}$ for the human operator be $y_{h} = z$.
Namely, the human operator can obtain the average position of the accessible robots $z$.
In \cite{HCF15}, position synchronization is theoretically proved 
in the absence of delays under a passivity assumption on
the human's decision process together with some additional assumptions.
However, inter-robot communication delays may destabilize the above control system.
This is why we present a system architecture with robustness against inter-robot communication delays.

\begin{remark}
The system architecture (\ref{eq:xiorig}), (\ref{eq:uriorig}) 
can achieve not only position synchronization but also synchronization of 
velocity at the same time, i.e.,
\begin{eqnarray}
\label{eq:goal2}
\lim_{t \to \infty}\|\dot{q}_{i} - v_{r}\| = 0,\
\lim_{t \to \infty}\|q_{i} - q_{j}\| = 0\quad  \forall{i, j} \in \V,
\end{eqnarray}
where $v_{r}$ is a constant reference velocity.
However, we omit this part because of space limitations.
\end{remark}

\section{Cooperative Control Architecture under Inter-Robot Communication Delays}
\label{sec:hrn}

\begin{figure}[t]
	\centering
	\includegraphics[width=75mm,clip]{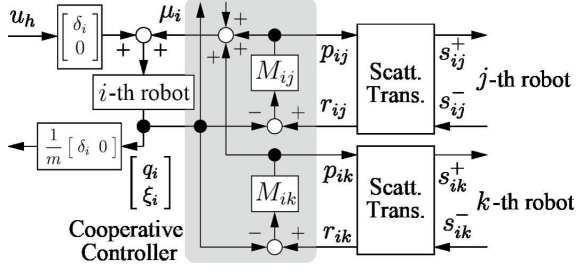}
	\caption{Block diagram of $i$-th robot and its cooperative controller including the 
	scattering transformation.}
	 \label{fig:roboti}
\end{figure}
\subsection{Passivity of Robot Dynamics}

In this section, we redesign the cooperative controller 
generating $u_{r,i}$ in order to guarantee stability and synchronization
despite the inter-robot communication delays.

In the presence of the delays, $i$-th robot cannot obtain $q_{j}(t)$ and $\xi_{j}(t)$ at time $t$.
We thus replace these signals 
by $r^{q}_{ij}, r^{\xi}_{ij}$, respectively.
Then, $i$-th robot's  dynamics is described by
\begin{align}
\dot{\xi}_{i} 
&= 
\sum_{j \in \N_{i}}b_{ij}(r^{q}_{ij} - q_{i}) \label{eq:xii} \\
\dot q_i
&= 
\sum_{j \in \N_{i}}a_{ij}(r^{q}_{ij} - q_{i}) 
- \sum_{j \in \N_{i}}b_{ij}(r^{\xi}_{ij} - \xi_{i}) + \delta_{i} u_{h} \label{eq:uri}
\end{align}
Let us now define $p_{ij} \in \R^{4}$ as
\begin{align}
\hspace{-2ex}
p_{ij} := M_{ij}
\left(
\left[
\begin{array}{cc}
\! r^{q}_{ij} \! \\
\! r^{\xi}_{ij} \!
\end{array}
\right]
- 
\left[
\begin{array}{cc}
\! q_{i} \! \\
\! \xi_{i} \!
\end{array}
\right]
\right)
=
M_{ij}
\left[
\begin{array}{cc}
\! r^{q}_{ij} - q_{i} \! \\
\! r^{\xi}_{ij} - \xi_{i} \!
\end{array}
\right], \label{eq:vij}
\end{align} 
where $M_{ij}
:= 
\left[
\begin{array}{cc}
\! a_{ij}I_{2} \!&\! -b_{ij}I_{2}\! \\
\! b_{ij}I_{2} \!&\! 0\!
\end{array}
\right]$ for all $j \in \N_{i}$ and $i$. 
It is easy to see from the definition that $M_{ij}$ is a passive map.

Then, (\ref{eq:xii}) and (\ref{eq:uri}) is described as a feedback system 
consisting of
\begin{align}
\label{eq:separate}
\left[
\begin{array}{c}
\dot{q}_{i} \\
\dot{\xi}_{i}
\end{array}
\right]
= \mu_i + 
\left[
\begin{array}{c}
\delta_{i} \\
0
\end{array}
\right] u_{h}
\end{align}
and the operation $\mu_i := \sum_{j \in \N_{i}} p_{ij}$.
Then, if there is no interaction with a human operator, 
i.e., $u_{h} \equiv 0$, and we define a storage function as
\begin{align}
\label{eq:Si}
S_{i}:=\frac{1}{2}\|q_{i}\|^{2}+\frac{1}{2}\|\xi_{i}\|^{2},
\end{align}
the system (\ref{eq:separate}) is passive from 
$\mu_i$ to $[q_{i}^{T}\ \xi^{T}_{i}]^{T}$.

From the above discussions, the system 
(\ref{eq:xii}) and (\ref{eq:uri}) can be 
reduced to a feedback interconnection of passive systems and 
a collection of passive maps $M_{ij}\ j\in {\mathcal N}_i$.
Accordingly, we can obtain the following inequality for each robot dynamics,
\begin{align}
\hspace{-1ex}
\dot{S}_{i} 
=& \left[
q_{i}^{T}\ \xi_{i}^{T}
\right] \sum_{j \in \N_{i}}p_{ij} \nonumber \\
=& 
-\!\! \sum_{j \in \N_{i}}a_{ij}\|q_{i} - r^{q}_{ij}\|^{2} 
\!+\! \sum_{j \in \N_{i}} r_{ij}^{T} p_{ij} 
\leq \sum_{j \in \N_{i}}r_{ij}^{T} p_{ij}. \label{eq:rijvij}
\end{align}
The above inequality implies that each robot dynamics is passive 
from the stacked vector of $r_{ij}$ to the stacked vector of $p_{ij}$ from all neighbor robots with respect to the storage function (\ref{eq:Si}).


%
%
\begin{figure}[t]
	\centering
	\includegraphics[width=87mm,clip]{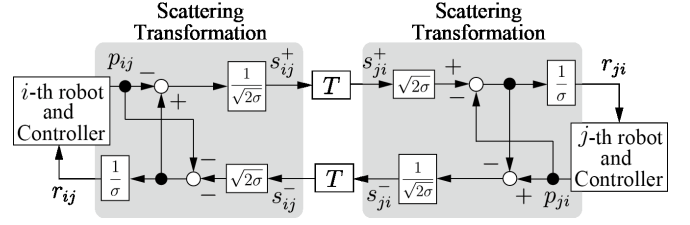}
	\caption{The scattering transformation between robot $i$ and robot $j$.}
	 \label{fig:scattering}
\end{figure}
%
%
\subsection{Scattering Transformation and Passivity of Communication Block}

In this subsection, we present a novel control architecture based on \cite{HCFS15}.
Hereafter, similarly to \cite{HCF15,HCYF17}, we assume that reference $q_r$ is constant, which means that
we are assuming that these signals are varying so slowly 
that the analyses under the constant reference are applicable to
the actual system.
Time varying reference is also addressed in \cite{HCYF17}, 
however its extension to the case with delays is left as a future work.

Following the architecture of \cite{HCFS15},
we let each agent exchange the variable $p_{ij}$ through scattering transformation instead 
of directly sending $q_{i}$ and $\xi_{i}$ as illustrated
in Fig. \ref{fig:roboti}.
In the present case, the scattering variables are defined as
\begin{align}
s^{+}_{ij} 
={}&
\frac{1}{\sqrt{2\sigma}}(-p_{ij} \!+\! \sigma r_{ij}), \
s^{-}_{ij} = \frac{1}{\sqrt{2\sigma}}(-p_{ij} \!-\! \sigma r_{ij}), \label{eq:s1} \\
s^{+}_{ji} 
={}&
\frac{1}{\sqrt{2\sigma}}(p_{ji} + \sigma r_{ji}),\hspace{2ex}
s^{-}_{ji} = \frac{1}{\sqrt{2\sigma}}(p_{ji} - \sigma r_{ji}), \label{eq:s2}
\end{align}
where $r_{ij} = [(r^{q}_{ij})^{T}\ (r^{\xi}_{ij})^{T}]^{T}$ 
and $\sigma > 0$ is constant.
The operation is illustrated in Fig. \ref{fig:scattering}.
It is immediate to see from this figure that
\begin{eqnarray}
s^{+}_{ji}(t) = s^{+}_{ij}(t - T) \label{eq:st1}, \ 
s^{-}_{ij}(t) = s^{-}_{ji}(t - T) \label{eq:st2},
\end{eqnarray}
where $T > 0$ is a delay in communication channel from $i$-th robot 
to $j$-th robot and it is assumed to be constant.
Hereafter, we suppose that both of $s^{+}_{ji}$ and $s^{-}_{ij}$
are equal to zero during the negative time $t < 0$ for all $i,j$.

We define $s_{q} := [(\sqrt{\sigma/2})q_{r}^{T}\ 0^{T}]^{T} \in \R^{4}$ and the storage function for the communication channel as
\begin{align*}
S^{c}_{ij} 
:=& 
\frac{1}{2} \int^{t}_{t - T} \|s^{+}_{ij} - s_{q}\|^{2} d\tau
+ \frac{1}{2} \int^{t}_{t - T} \|s^{-}_{ji} + s_{q}\|^{2} d\tau.
\end{align*}
Thus, if we define $\bar{r}^{q}_{ij} := r^{q}_{ij} - q_{r}$ and 
$\bar{r}_{ij} := [(\bar{r}^{q}_{ij})^{T}\ (r^{\xi}_{ij})^{T}]^{T}$, 
the time derivative of $S^{c}_{ij}$ is given by
\begin{align}
\label{eq:Vijdot}
\dot{S}^{c}_{ij} 
=& 
\frac{1}{4\sigma} (\|-p_{ij} + \sigma \bar{r}_{ij}\|^{2} - \|p_{ji} + \sigma \bar{r}_{ji}\|^{2} 
 + \|p_{ji} - \sigma \bar{r}_{ji}\|^{2} \nonumber \\
& - \|-p_{ij} - \sigma \bar{r}_{ij}\|^{2})  
= - p_{ij}^{T}\bar{r}_{ij} - p_{ji}^{T}\bar{r}_{ji}.
\end{align}
Thus, if we interpret $[-p_{ij}^{T}\ -p_{ji}^{T}]^{T}$ and 
$[\bar{r}_{ij}^{T}\ \bar{r}_{ji}^{T}]^{T}$ as the input and output respectively,
the communication channel becomes passive with respect to $S^{c}_{ij}$.

\subsection{Passivity of Robotic Network}

In this subsection, we show passivity of the robotic network discussed above.
First, we define the signals $\bar{q}_{i} := q_{i} - q_{r}$, for all $i$
and $\bar{z} := z - q_{r}$.
Then, the error dynamics on (\ref{eq:xii}) and (\ref{eq:uri})
 is given by
\begin{align}
\dot{\bar{q}}_{i} 
&= 
\sum_{j \in \N_{i}}a_{ij}(\bar{r}^{q}_{ij} - \bar{q}_{i}) 
- \sum_{j \in \N_{i}}b_{ij}(r^{\xi}_{ij} - \xi_{i}) + \delta_{i} u_{h}\label{eq:errorqi}\\
\dot{\xi}_{i} 
&= 
\sum_{j \in \N_{i}}b_{ij}(\bar{r}^{q}_{ij} - \bar{q}_{i}), \label{eq:errorxii}
\end{align}
where (\ref{eq:errorqi}) is hold because $q_{r}$ is constant.
From (\ref{eq:errorqi}) and (\ref{eq:errorxii}), 
(\ref{eq:separate}) is rewritten as
\begin{align}
\label{eq:errorseparate}
\left[
\begin{array}{c}
\dot{\bar{q}}_{i} \\
\dot{\xi}_{i}
\end{array}
\right]
= 
\mu_{i} 
+ 
\left[
\begin{array}{c}
\delta_{i} \\
0
\end{array}
\right] u_{h}.
\end{align}
Then, we define the storage function for the error dynamics as
\begin{align}
\bar{S}_{i} 
:=
\frac{1}{2}\|\bar{q}_{i}\|^{2} + \frac{1}{2}\|\xi_{i}\|^{2},\ i \in \V.
\end{align}
From (\ref{eq:rijvij}), the time derivative of $\bar{S}_{i}$ 
along the trajectory of (\ref{eq:errorseparate}) is given by
\begin{align}
\dot{\bar{S}}_{i} 
={}& \left[
\bar{q}_{i}^{T}\
\xi_{i}^{T}
\right] \sum_{j \in \N_{i}}p_{ij} 
+ \bar{q}_{i}^{T} \delta_{i} u_{h} \nonumber \\
={}& 
- \sum_{j \in \N_{i}}a_{ij}\|q_{i} - r^{q}_{ij}\|^{2}
+ \sum_{j \in \N_{i}}
\bar{r}_{ij}^{T} p_{ij}
+ \bar{q}_{i}^{T} \delta_{i} u_{h}. \label{eq:Sidot}
\end{align}
Hence, we define a total storage function for the robotic network as
\begin{eqnarray}
\label{eq:totalstrfunc}
S 
:= 
\frac{1}{m}\sum_{i \in \V}\bar{S}_{i}
+ \frac{1}{m}\sum_{(i, j) \in \E} S^{c}_{ij}.
\end{eqnarray}
From (\ref{eq:Sidot}) and (\ref{eq:Vijdot}), 
the time derivative of $S$ is given by
\begin{align*}
\dot{S}
=& 
\frac{1}{m}\sum_{i \in \V} \dot{\bar{S}}_{i} + 
\frac{1}{m}\sum_{(i,j) \in \E} \dot{S}^{c}_{ij} \nonumber \\
=& 
- \frac{1}{m}\sum_{i \in \V}\sum_{j \in \N_{i}}a_{ij}\|q_{i} - r^{q}_{ij}\|^{2} 
+ \bar{z}^{T}u_{h} \leq \bar{z}^{T}u_{h}.
\end{align*}
Thus, the robotic network is passive from human input $u_{h}$ 
to the average position of accessible robots $\bar{z}$ 
with respect to (\ref{eq:totalstrfunc}).
In the next section, 
we close the robotic network with a human operator based on 
this passive property.

\begin{figure}[t]
	\centering
	\includegraphics[width=70mm,clip]{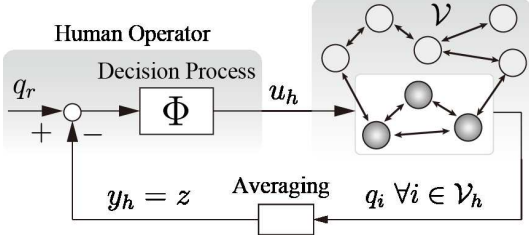}
	\caption{Block diagram of the human-robotic network system.}
	\label{fig:blockposition}
\end{figure}
%
%
%
\section{Interconnected system of Robotic Network and Passive Human Operator}
\label{sec:main}

In this section, we prove position synchronization (\ref{eq:goal1})
by utilizing passivity.

\subsection{Synchronization to Human's Desired Position}

First, we start by the description of a human operator.
Similarly to \cite{HCF15,HCYF17}, we assume that the operator determines
the command $u_{h}$ based on the error $q_{r} - y_{h}$ between the 
reference $q_{r}$ and the feedback information $y_{h} = z$.
Then, the entire system is illustrated as Fig. \ref{fig:blockposition},
where the human decision process is denoted by $\Phi$.
Then, we assume that a human operator's decision process $\Phi$ 
is input strictly passive from $q_{r} - z$ to 
$u_{h}$, i.e., there exists a storage function 
$S_h^{\Phi}$ and $\epsilon > 0$ such that
\begin{align}
\dot S_h^{\Phi} \leq
(q_{r} \!-\! z)^{T}u_{h}
- \epsilon \|q_{r} \!-\! z\|^{2}.
\end{align}
The validity of this assumption is examined in \cite{HCF15,HCYF17} 
using experimental data obtained from a human-in-the-loop simulator,
where it is confirmed that the statement is true over a prescribed frequency domain.
We need to analyze the effect of the communication delays on the human property
but this exceeds the scope of this paper and we leave the issue as a future work.


From this human passivity assumption, we have the following result.
\begin{proposition}
\label{the:1}
Consider the system (\ref{eq:errorqi}) and (\ref{eq:errorxii}),
and the scattering transformation (\ref{eq:s1}) and (\ref{eq:s2}) 
for all $j \in \N_{i}$ and all $i$.
Then, the feedback system achieves the condition (\ref{eq:goal1})
if the communication graph is fixed, undirected and connected, 
and the human decision process $\Phi$ is input strictly passive.
\end{proposition}
\begin{proof}
First, we define an energy function $U$ as
\begin{align}
U
:= S +
\int^{t}_{0}\!
\left(
-\bar{z}(\tau)u_{h}(\tau) \!-\! \epsilon \|\bar{z}(\tau)\|^{2} 
\right) d\tau + \beta,
\end{align}
where $\beta$ is a positive constant.
From the passivity of robotic network shown in Section \ref{sec:hrn} 
and human operator's passivity assumed above, we obtain
\begin{align}
\dot{U}
&=
- \frac{1}{m}\sum_{i \in \V}\sum_{j \in \N_{i}}a_{ij}\|q_{i} - r^{q}_{ij}\|^{2} 
+ \bar{z}^{T}u_{h} - \bar{z}^{T}u_{h} - \epsilon \|\bar{z}\|^{2} \nonumber \\
&= 
- \frac{1}{m}\sum_{i \in \V}\sum_{j \in \N_{i}}a_{ij}\|q_{i} - r^{q}_{ij}\|^{2} 
- \epsilon \|\bar{z}\|^{2} \leq 0. \label{eq:deriv}
\end{align}
Thus, it is guaranteed that 
the all states are bounded in spite of communication delays.

Let us define $x^{i}_{t}$ such that 
$x^{i}_{t}(\theta) = x_{i}(t+\theta)$ for $\theta \in [-T, 0]$.
Then, the LaSalle's invariance principle for time delay systems \cite{HL93}
is applicable and any solution $x_{t}$ of the system converges to
the largest invariant set in the set of functions 
satisfying $\dot{U} \equiv 0$. 
Thus, $\dot{U} = 0$ means 
$q_{i} \equiv r^{q}_{ij}\ (i, j) \in \E$ 
and $\bar{z} \equiv 0$.
Hence, we can conclude as follows.
\begin{align}
&\lim_{t \to \infty}(r^{q}_{ij} - q_{i})
= 0\quad (i, j) \in \E \label{eq:rqconv} \\
&\lim_{t \to \infty}\bar{z} = 0 \label{eq:zqconv}
\end{align}
Next, in order to analyze further, we need to follow the 
behavior of $\bar{r}_{ij}$.
First, we define $e_{ij}$ and $e_{ji}$ as

\begin{align*}
e_{ij}
:=& 
- \frac{1}{\sigma}\{(r^{q}_{ij}-q_{i}) + (r^{q}_{ji}(t-T)-q_{j}(t-T))\}, \\
e_{ji}
:=& 
- \frac{1}{\sigma}\{(r^{q}_{ji}-q_{j}) + (r^{q}_{ij}(t-T)-q_{i}(t-T))\}.
\end{align*}
Then, each element of (\ref{eq:s1})--(\ref{eq:st2}) can be given as
\begin{align}
r^{q}_{ij} 
={}& 
r^{q}_{ji}(t-T) 
+ \frac{1}{\sigma}
\{ b_{ij}(r^{\xi}_{ij}-\xi_{i}) \nonumber \\
& +  b_{ji}(r^{\xi}_{ji}(t-T)-\xi_{j}(t-T))\}
+ e_{ij},  \label{eq:rqij} \\
r^{q}_{ji} 
={}& 
r^{q}_{ij}(t-T) 
+ \frac{1}{\sigma}
\{ b_{ji}(r^{\xi}_{ji}-\xi_{j}) \nonumber \\
& + b_{ij}(r^{\xi}_{ij}(t-T)-\xi_{i}(t-T))\}
+ e_{ji}.  \label{eq:rqji} \\
r^{\xi}_{ij} 
={}& 
r^{\xi}_{ji}(t-T) 
+ b_{ij}e_{ij}, \\
r^{\xi}_{ji} 
={}&
r^{\xi}_{ij}(t-T) 
+ b_{ji}e_{ji}, \label{eq:rxiij} 
\end{align}
Furthermore, subtracting (\ref{eq:rqji}) at time $t - T$ 
from (\ref{eq:rqij}) yields
\begin{align}
r^{q}_{ij} + r^{q}_{ij}(t-2T) 
={}&
2 r^{q}_{ji}(t-T)
- \frac{b_{ij}}{\sigma}\{ (r^{\xi}_{ij}-\xi_{i}) \nonumber \\
& - (r^{\xi}_{ij}(t \!-\! 2T)\!-\! \xi_{i}(t \!-\! 2T))\} \nonumber \\
& - a_{ij}e_{ij} + a_{ji}e_{ji}(t-T).
\label{eq:rijrij}
\end{align}
Utilizing the equations in (\ref{eq:rxiij}), we have
$r^{\xi}_{ij} 
= 
r^{\xi}_{ij}(t-2T) 
+ e_{ij} + e_{ji}(t-T)$.
Now, according to (\ref{eq:rqconv}), the signals $e_{ij}$ and $e_{ji}$ 
converge to $0$, i.e., 
\begin{eqnarray}
\label{eq:dijdji}
\lim_{t \to \infty}e_{ij} = \lim_{t \to \infty}e_{ji} = 0 
\end{eqnarray}
holds. 
Thus, taking the limit of (\ref{eq:rijrij}) yields
\begin{align}
\lim_{t \to \infty}(r^{q}_{ij} + r^{q}_{ij}(t-2T) - 2 r^{q}_{ji}(t-T)) = 0.
\label{eq:rijrij2}
\end{align}
The same equation holds for $k \in \N_{i}$ as
\begin{align}
\lim_{t \to \infty}(r^{q}_{ik} + r^{q}_{ik}(t-2T) - 2 r^{q}_{ki}(t-T)) = 0.
\label{eq:rikrik2}
\end{align}
On the other hand, because the signal $r_{ij}^{q}$ for all $j \in \N_{i}$ converges  
to $q_{i}$ from (\ref{eq:rqconv}), we obtain
\begin{align}
\label{eq:rijrik}
\lim_{t \to \infty}(r^{q}_{ij} - r^{q}_{ik}) = 0\quad \forall{j, k} \in \N_{i}.
\end{align}
Subtracting (\ref{eq:rikrik2}) from (\ref{eq:rijrij2}) 
and using (\ref{eq:rijrik}),
we have
\begin{align}
\label{eq:rjirki}
\lim_{t \to \infty}(r^{q}_{ji} - r^{q}_{ki}) = 0.
\end{align}
Because of (\ref{eq:rqconv}), equation (\ref{eq:rjirki}) implies $\lim_{t \to \infty}(q_{j} - q_{k}) = 0$, which holds for $(j, k) \in \E$.
Then, we have $\lim_{t \to \infty}(q_{i} - q_{j}) = 0$ for all $i, j \in \V$.
Furthermore, from (\ref{eq:zqconv}), we have
\begin{align}
\lim_{t \to \infty}\bar{z}
=
\lim_{t \to \infty} \left(
\frac{1}{m}\sum_{i \in \V_{h}}q_{i} - q_{r} 
\right) = 0.
\end{align}
Therefore, we can conclude that $\lim_{t \to \infty}(q_{i} - q_{r}) = 0$ 
for all $i \in \V$.
This completes the proof.
\end{proof}
\begin{remark}
As in the same way, the velocity synchronization (\ref{eq:goal2}) is achieved 
with the same control architecture.
However, we omit the result because of space limitations. 
\end{remark}

\begin{figure}[t]
	\centering
	\includegraphics[width=87mm,clip]{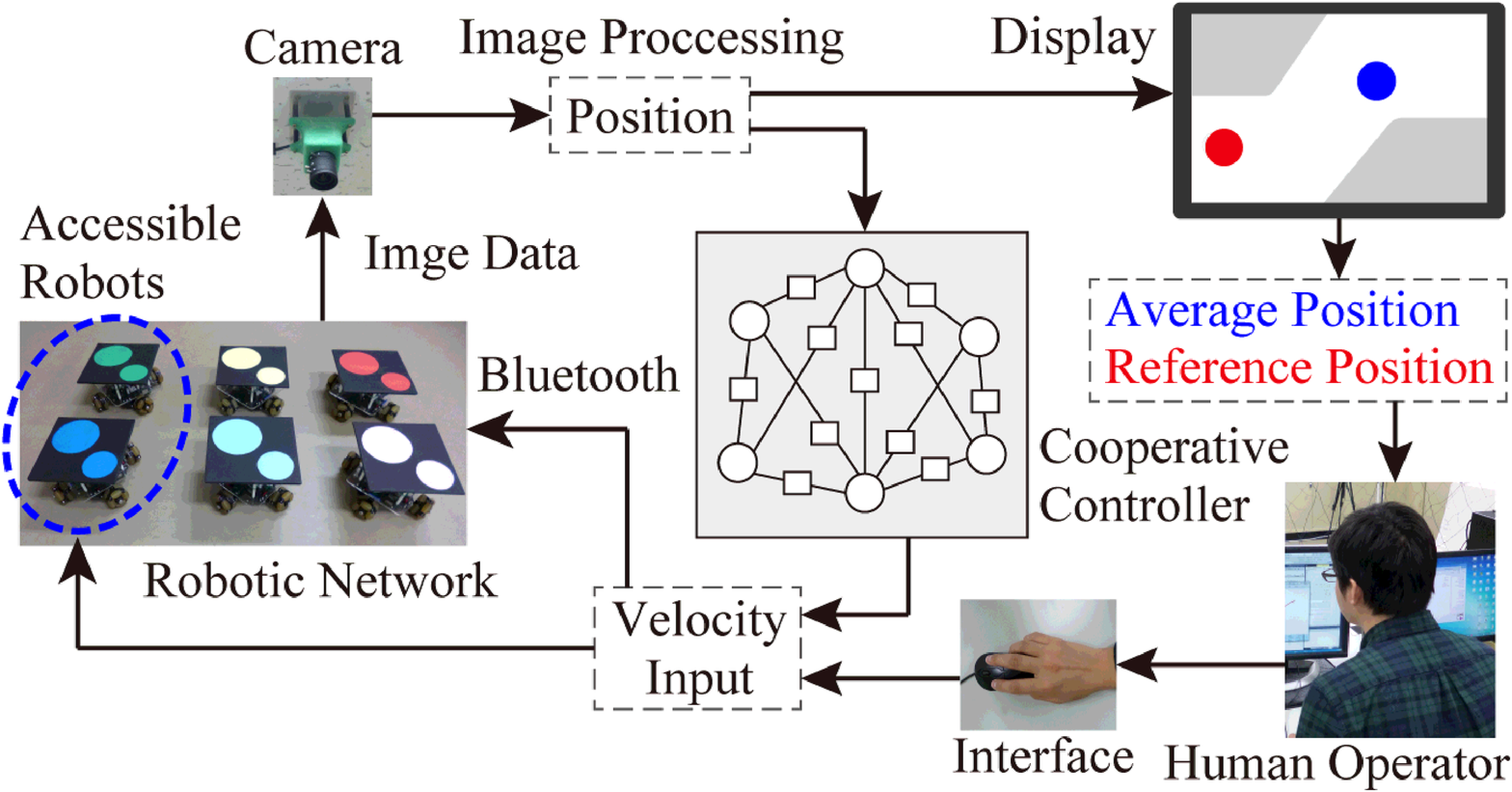}
	\caption{The architecture overview of experiment system.}
	 \label{fig:expSys}
\vspace{2ex}
	\centering
	\includegraphics[width=70mm,clip]{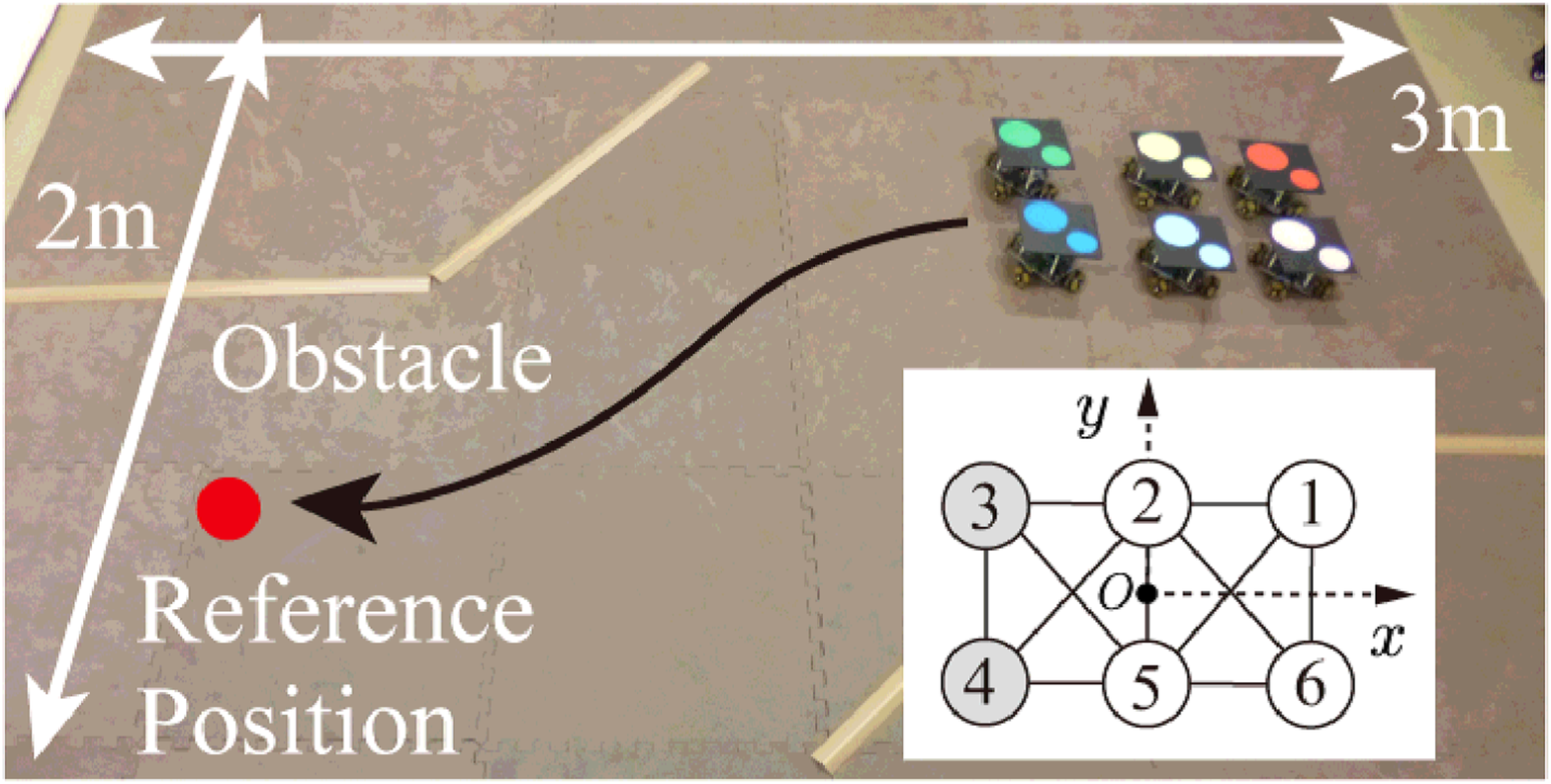}
	\caption{The robotic network and obstacles on the field.
	The communication graph is shown in right bottom.
	The accessible robots are in grey, i.e., 3rd and 4th robots.}
	 \label{fig:expField}
\end{figure}

%
\section{Experiment}
\label{sec:experiment}

In this section, we show the experiment results of the discussed control system.
In this paper, we focus on the influence of the communication delay
on the human operator's behavior.  
For that purpose, we conduct two types of experiments, without inter-robot communication delay and  
with inter-robot communication delay.
Then, we investigate the differences by comparing the robots' trajectories and human input.
Although it is also important to investigate the human passivity assumption, 
as already mentioned in Section \ref{sec:main},
we leave this as a future work.

\begin{figure}[t]
   \centering
 	\includegraphics[width=55mm]{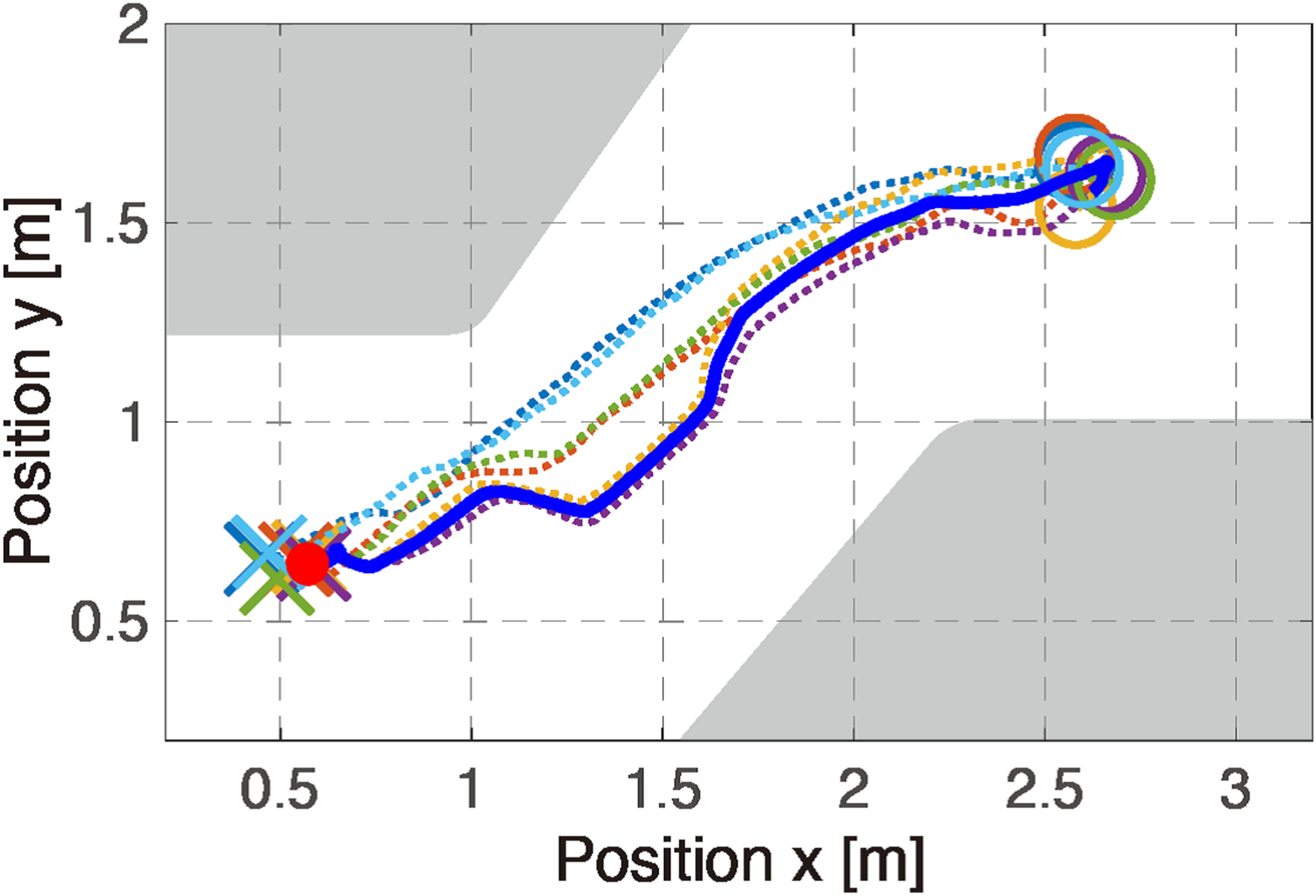}
 \begin{minipage}{0.49\hsize}
  \centering
   \includegraphics[width=47mm]{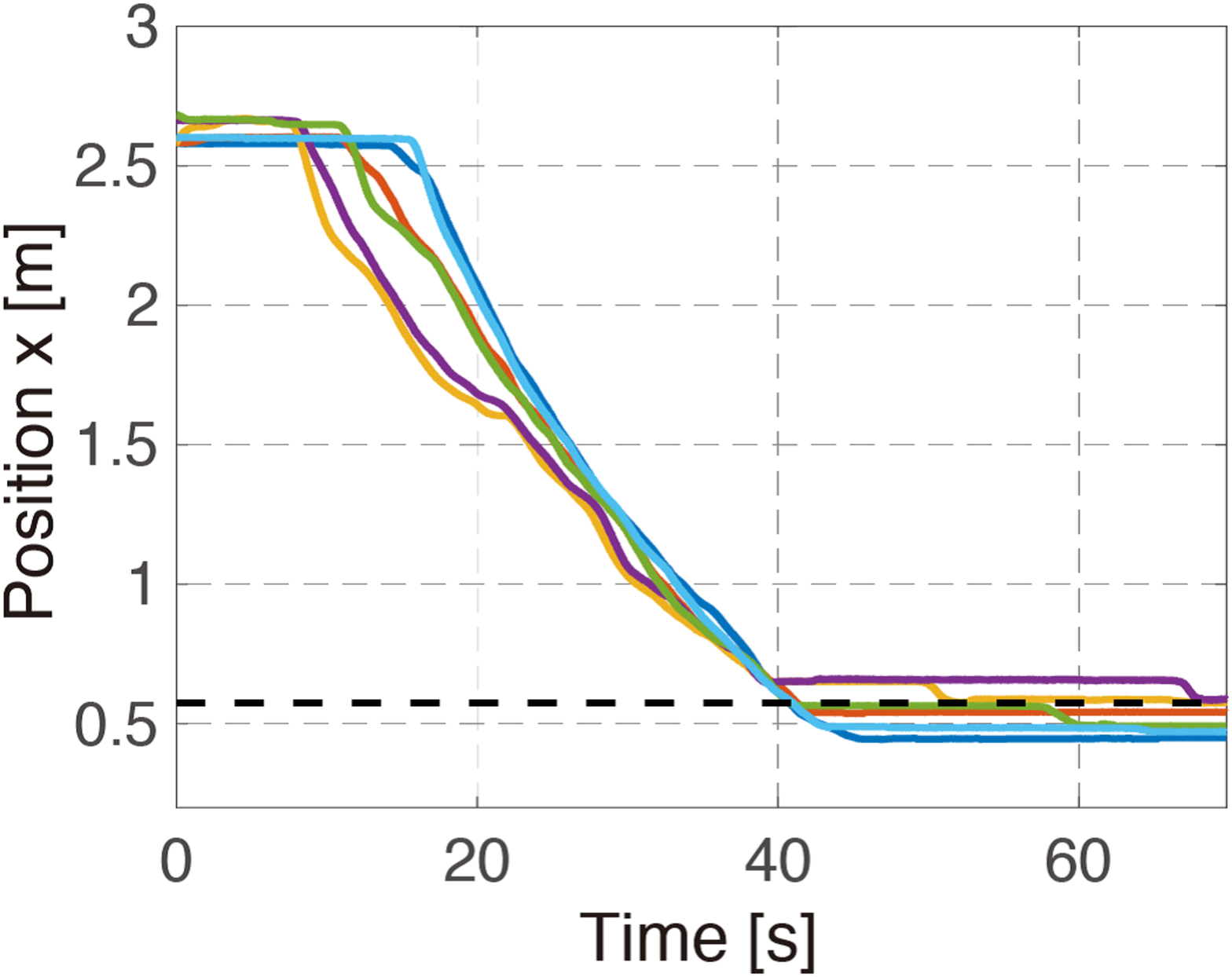}
 \end{minipage}
 \begin{minipage}{0.49\hsize}
  \centering
   \includegraphics[width=47mm]{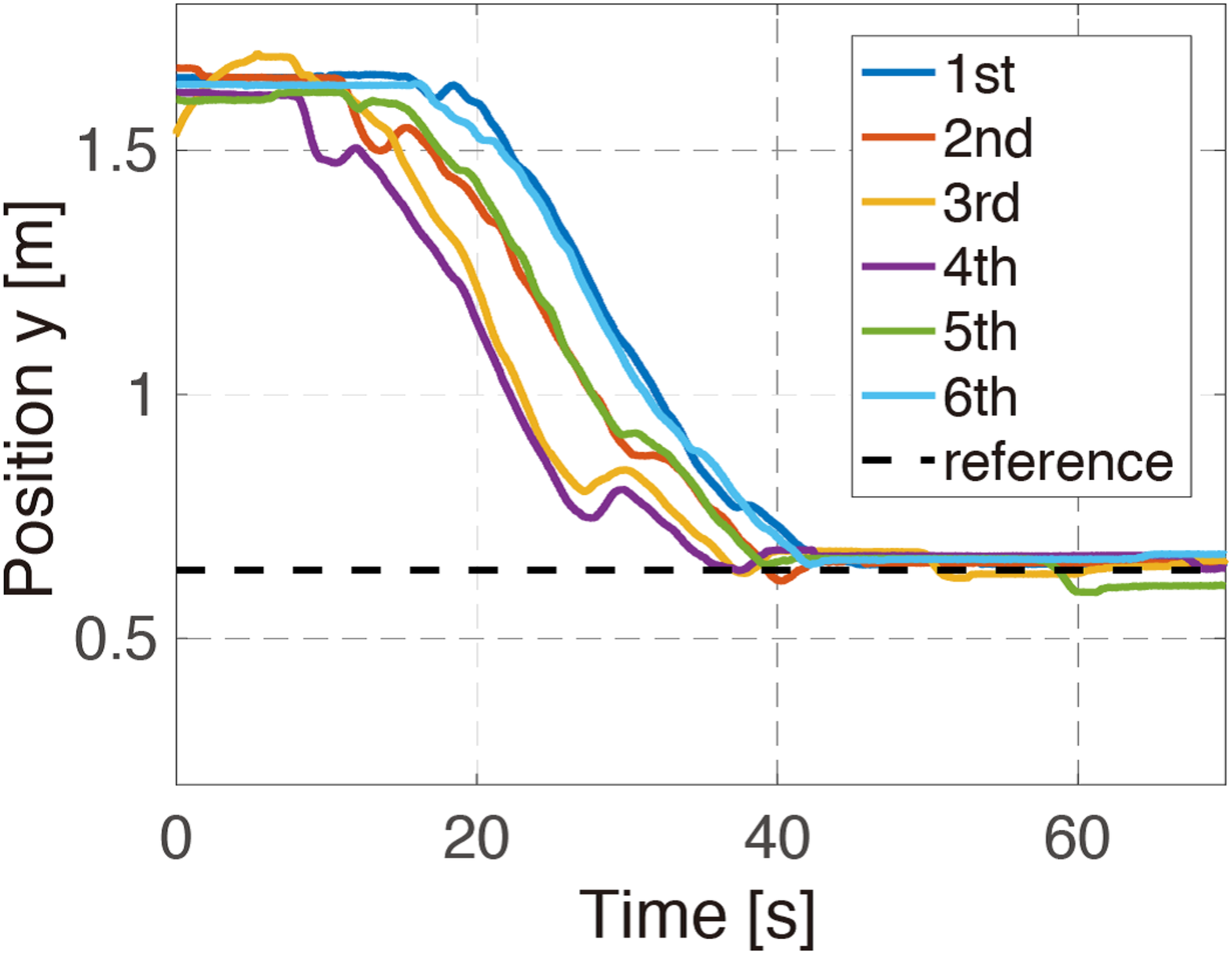}
 \end{minipage}
   \centering
 	\includegraphics[width=67mm]{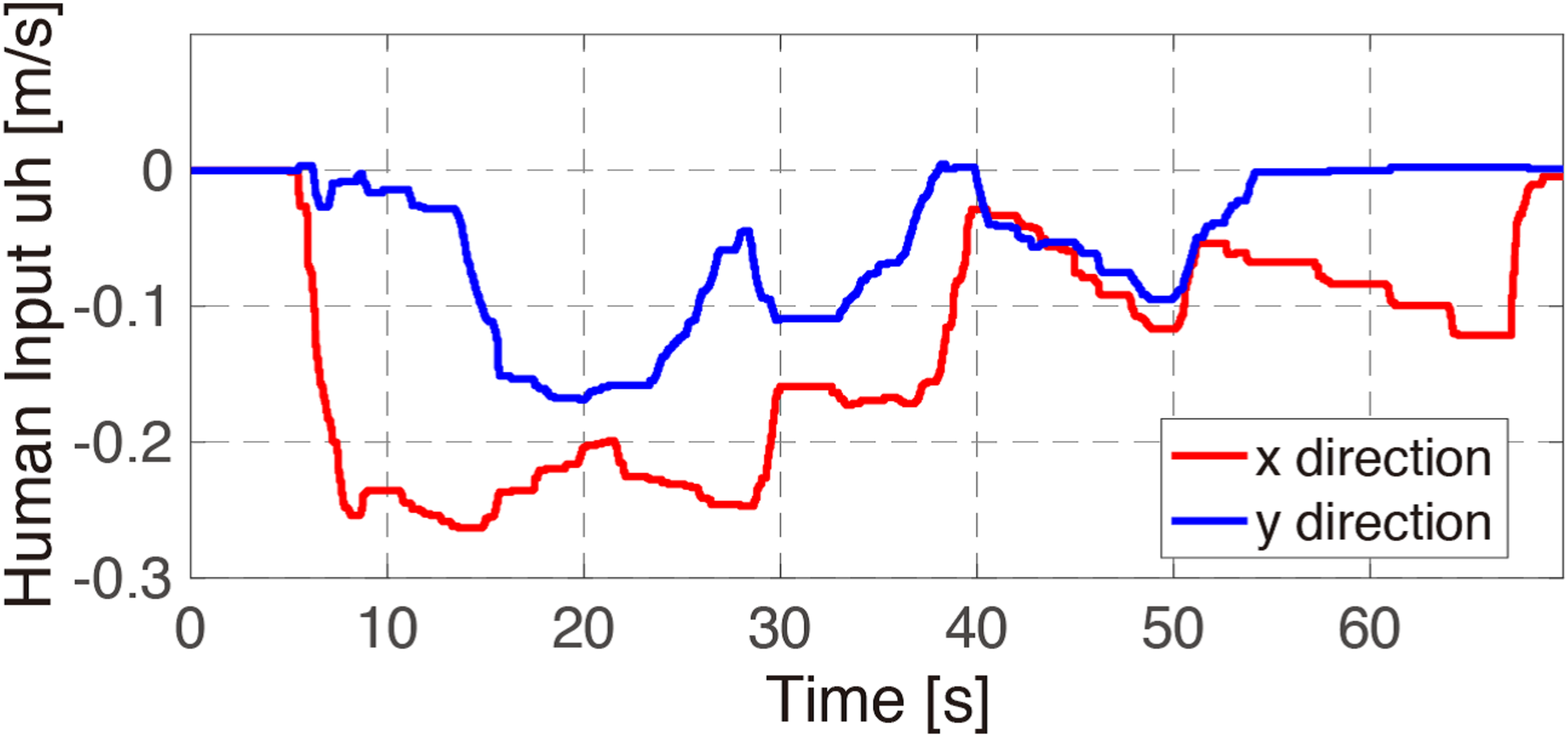}
   \caption{The experiment result without inter-robot communication delay.
   			(top) The trajectories of all robots. 
			The average position of accessible robots are shown by blue curve,
			and each robot's initial and last positions are shown by circle 
			and cross, respectively.
			The reference position is shown by red dot.
			(middle) The trajectories in x-axis and y-axis.
			(bottom) Human input.}
  \label{fig:expResult1}
\end{figure}
\begin{figure}[t]
   \centering
 	\includegraphics[width=55mm]{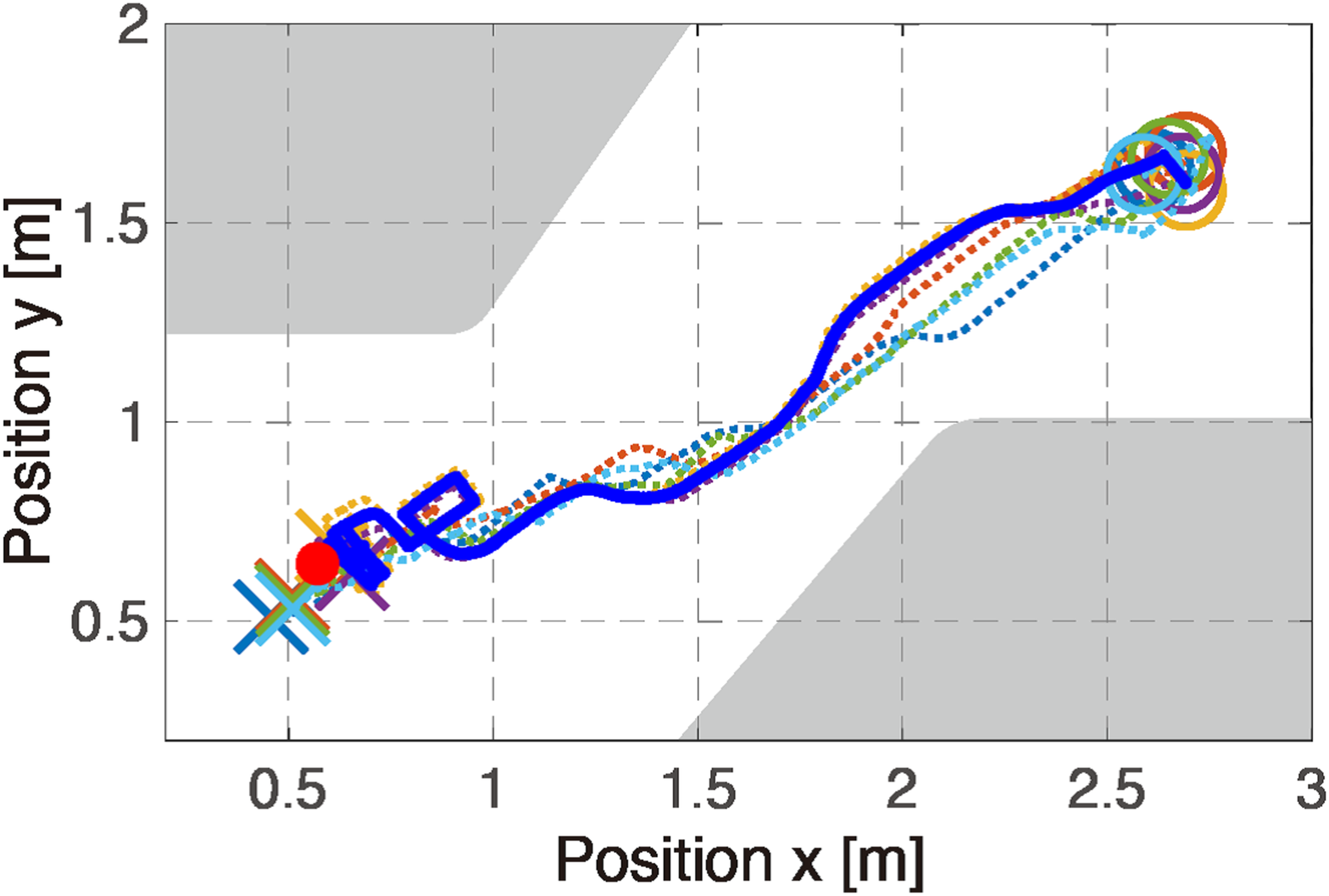}
 \begin{minipage}{0.49\hsize}
  \centering
   \includegraphics[width=47mm]{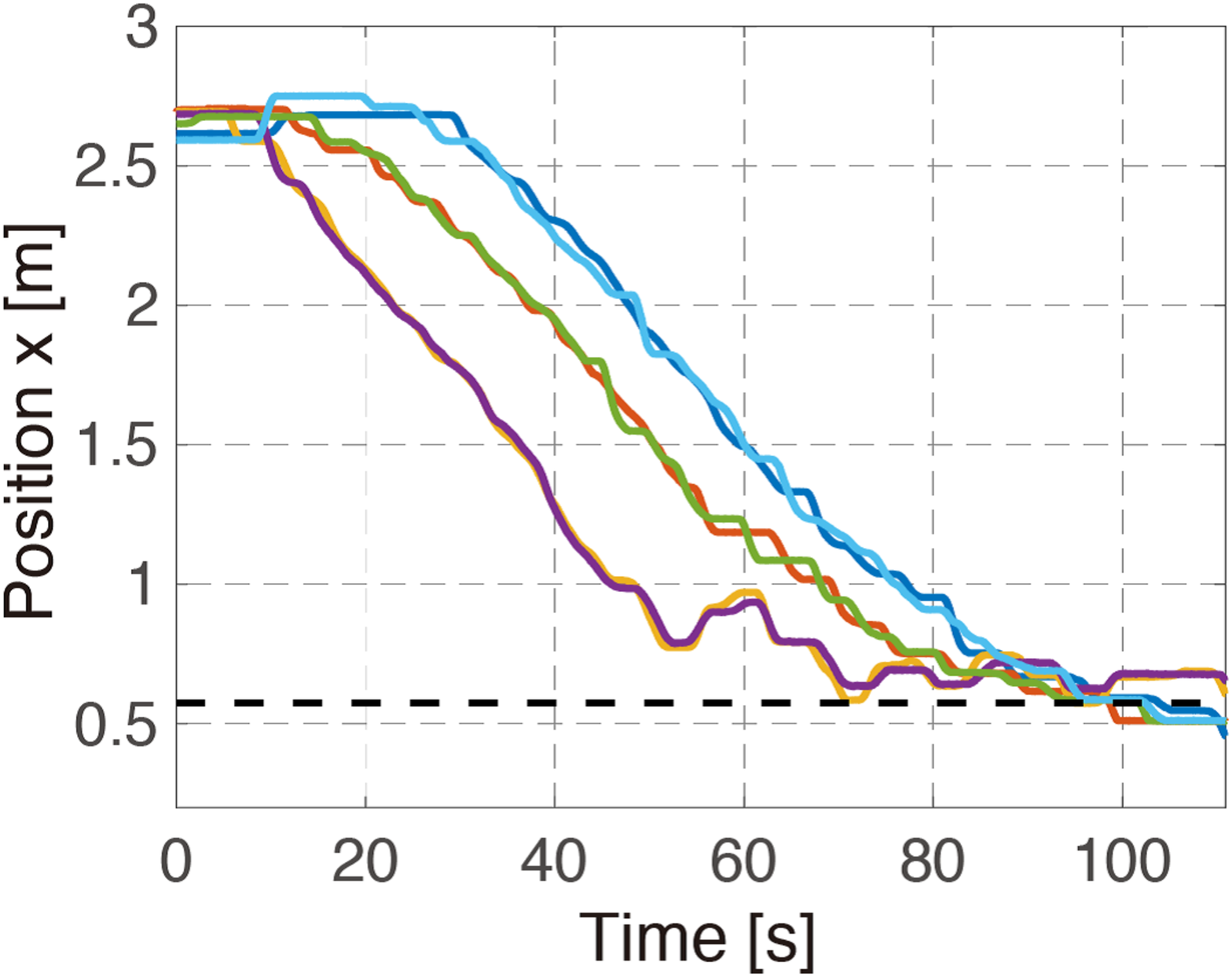}
 \end{minipage}
 \begin{minipage}{0.49\hsize}
  \centering
   \includegraphics[width=47mm]{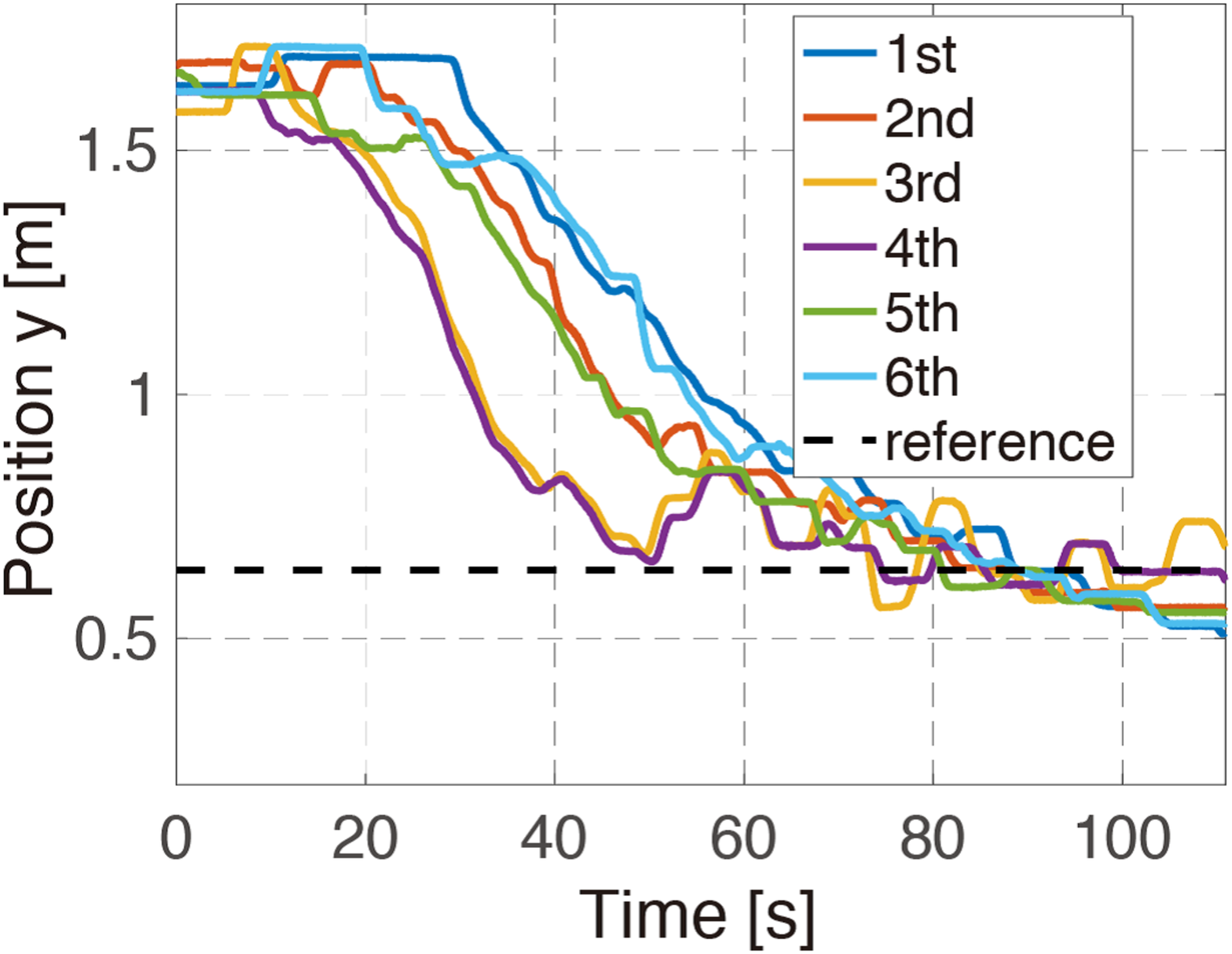}
 \end{minipage}
   \centering
 	\includegraphics[width=67mm]{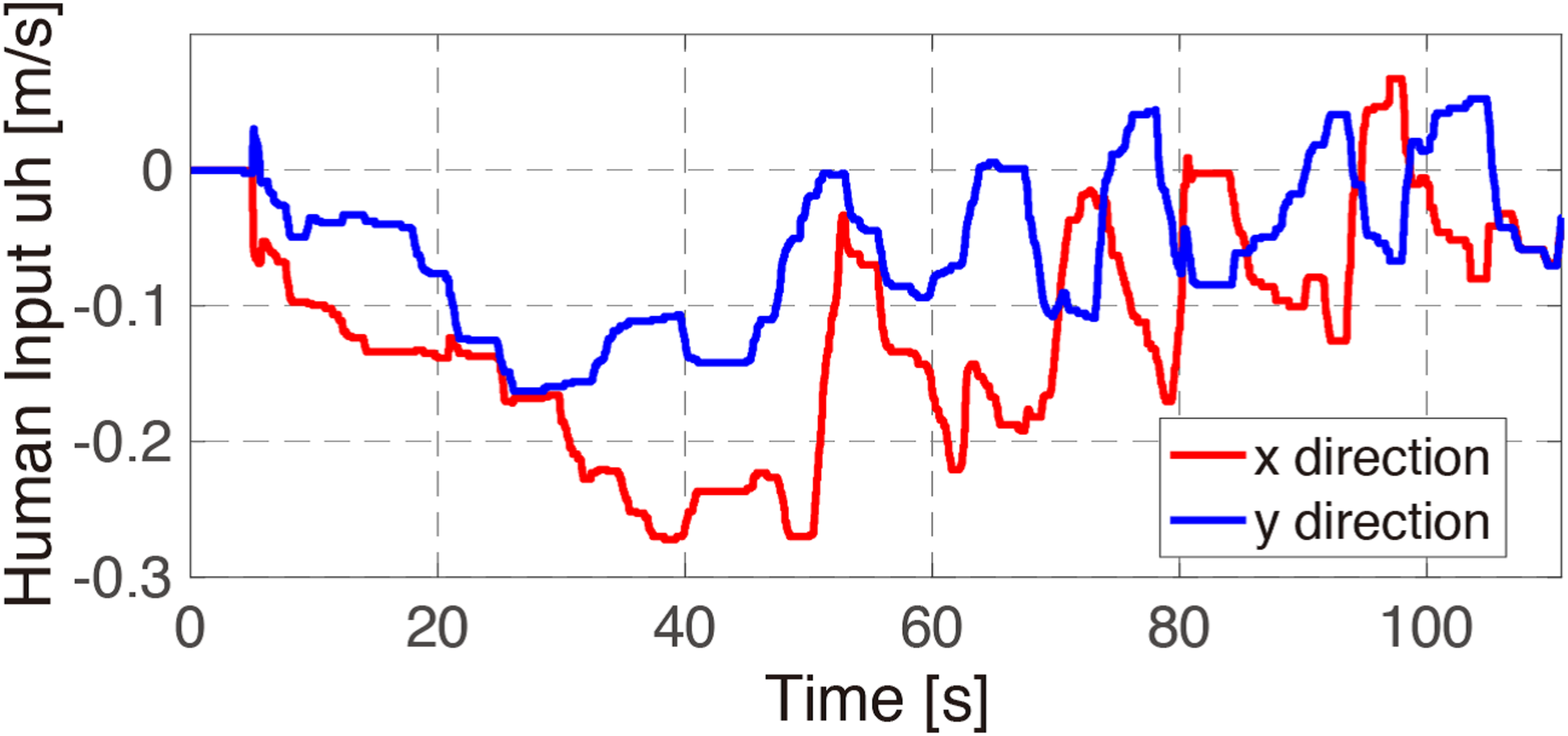}
	\caption{The experiment result with inter-robot communication delay and 
	scattering transformation.
	(top) The trajectories of all robots. 
	The average position of accessible robots are shown by blue curve,
	and each robot's initial and last positions are shown by circle 
	and cross, respectively.
	The reference position is shown by red dot.
	(middle) The trajectories in x-axis and y-axis.
	(bottom) Human input.}
  \label{fig:expResult2}
\end{figure}
\subsection{Experimental Setting}
In this subsection, we introduce our experiment system and 
the intended scenario.
The architecture of experiment system is illustrated in Fig. \ref{fig:expSys}.
We use 6 omnidirectional robots named TDO48 from TOSADENSHI Inc.
and obtain image data 
by using a ceiling camera named Firefly MV from Point Grey. 
From image processing by C++, we obtain all robots positions.
Even though the control architecture is implementable by each robot processor, 
we artificially implement all cooperative controller 
in a computer for simplicity.
The inter-robot communication delays are also generated in Simulink  
and the delays set 0.5s.
Simulink is running on dSPACE for the real time implementation. 
On the other hand, we show the average position
of accessible robots and given reference position via GUI programmed by C++.
Then, the participant inputs velocity command through mouse.
Finally, the velocity inputs are converted to motor angular velocities, 
and then, transmitted to each robot via Bluetooth.

For the scenario, we ask the participant to control robotic network to a given position,
as shown by red dot in Fig. \ref{fig:expField}.
The reference position is $q_{r} = [0.55\ 0.60]^{T}$.
Furthermore, we add the obstacles which the participant have 
to avoid collision with them.
In order to avoid the collisions between each robot, 
we utilize the biases $d_{i}\ \forall{i} \in \V$ 
and denote real position $\eta_{i}$ as 
$\eta_{i} = q_{i} + d_{i}\ \forall{i} \in \V$.
The biases are given as
$d_{1} = [0.35\ 0.175]^{T}, d_{2} = [0\ 0.175]^{T}, d_{3} = [-0.35\ 0.175]^{T}, 
d_{4} = [-0.35\ -0.175]^{T}, d_{5} = [0\ -0.175]^{T},  
d_{6} = [0.35\ -0.175]^{T}$.
During this section, we show all figures using the biased positions,
rather than the real positions $\eta_{i}$.
The communication graph of the robotic network is shown in Fig. \ref{fig:expField}. 
The initial positions of all robots are $q_{i}(0) = [2.6\ 1.6]^{T}$ and 
$\xi_{i}(0) = 0\ \forall{i} \in \V$.
The parameters are set as $a_{ij} = 0.2,\ b_{ij} = 0.05$ and 
$\sigma = 1.0$.


\subsection{Experimental Results}

First, the result without inter-robot communication delay is shown
in Fig. \ref{fig:expResult1}.
The human operator successfully control the accessible robots, 
and all robots reached the target position.
Note that non-accessible robots could close the distances with the accessible robots around
30s while the accessible robots are still moving.
This shows one of the advantage in exchanging the integral value, $\xi_{i}$,
of $-\sum_{j \in \N_{i}}b_{ij}(q_{i} - q_{j})$.
The human input is shown in the bottom part of Fig. \ref{fig:expResult1}.
The result shows that the human operator smoothly controlled the robotic network.

Next, we show the result with inter-robot communication delay and 
the scattering transformation in Fig. \ref{fig:expResult2}.
In the same way as the delay free case, 
all robots reached the target position.
However, comparing the trajectories with the result without delay,
non-accessible robots couldn't close the distances with the accessible robots.
As a result, the accessible robots moved back and forth around 
the reference position.
The arrival time is delayed by about 70s.
The human input is shown in the bottom of Fig. \ref{fig:expResult2}.
After 50s, the human operator repeated adjustments 
because non-accessible robots followed late.
Comparing to the result of delay free case, 
the human input looks more fluctuated,
which can be interpreted as the deterioration of human's
operability.

In summary, position synchronization by the proposed architecture is verified
through this experiment.
In addition to this result, 
since we observed performance degradation of human, 
thorough investigation is needed regarding human performance.
To this end, we need a criteria to measure the performance qualitatively.

%
\section{Conclusions}
\label{sec:conclusions}

In this paper, we have investigated a cooperative control architecture 
of human-robotic networks in the presence of inter-robot communication delays, 
where the objective is to guarantee position synchronization 
to human's desired position. 
First, we proposed control architecture based on \cite{HCF15,HCYF17} 
and the scattering transformation. 
Then, we showed passivity of the robotic network. 
Next, by using human passivity assumption, 
we showed that the feedback system achieves position synchronization. 
Finally, we demonstrated the efficiency of the proposed architecture through experiments, which managed to move all robots to target position 
in the same way as delay free case.
Furthermore, we investigated the influences of inter-robot communication delays
on human's operability.

\addtolength{\textheight}{-12cm}   



%

%



\end{document}